\documentclass[12pt]{article}
\usepackage{graphicx}
\usepackage{url}
\usepackage{graphics, bm,amsmath, amsfonts, amssymb, subfigure, amstext}
\usepackage{amsthm,wrapfig, amsopn, latexsym,  epsfig,multicol, color}

\usepackage{times}
\newtheorem{definition}{Definition}
\newtheorem{hypothesis}{Conjecture}
\newenvironment{rem}{\textbf{Remark \arabic{section}.%
\addtocounter{rem}{1}\arabic{rem}.}}{\hfill $\blacksquare$}
\newcounter{rem}[section]

\newcommand{\df}{\stackrel{\textrm{def}}{=}}
\newtheorem{cor}{Corrolary}[section]
\newtheorem{theorem}{Theorem}[section]
\newcommand{\eqdef}{\overset{def}{=}}

\def\maxl{\max\limits}
\def\minl{\min\limits}
\def\supl{\sup\limits}

\title{Epsilon-complexity of continuous functions}
\author{Boris Darkhovsky, Alexandra Pyriatinska}
\date{}
\begin{document}
\maketitle

\begin{abstract}
A formal definition of   $\epsilon$-complexity of an individual  continuous function defined on a unit cube is proposed. This definition is consistent with the   Kolmogorov's idea  of  the complexity of an object.
 A definition of  $\epsilon$-complexity for  a class of continuous functions with a given modulus of continuity is also proposed. Additionally, an explicit formula for  the $\epsilon$-complexity of a functional class is   obtained. As a consequence, the paper finds that the $\epsilon$-complexity for the H\" older class of functions can be characterized by a pair of real numbers. Based on these results the papers formulates a conjecture  concerning the $\epsilon$-complexity of an individual  function  from the H\"older class. We also propose a conjecture about characterization of  $\epsilon$-complexity  of a  function  from the H\"older class given on a discrete grid.
\end{abstract}

Keywords: continuous functions, complexity.


\section{Introduction}
The concept of   "complexity" of  an object is one of the fundamental scientific  paradigms.  There are numerous  efforts   in the literature to  define  the  complexity properly. There are many  attempts
to apply it in practice as well.

One of the first  efforts to provide the  quantitative approach to  the concept of   "complexity of a  physical system" was made in 1870s by an Austrian  physicist Ludwig Boltzmann who  had  introduced the notion of   entropy in  equilibrium statistical physics. The greater the  entropy, the more "complicated"  the system is.

In  1940s Claude Shannon developed   the  information theory using  the   concept of  entropy of a probability distribution. He interpreted the entropy as  a measure of the "degree of uncertainty" which is peculiar to a particular probability distribution.  Under natural conditions  this  measure  was proven to be unique. It is known that the number $N$ of "typical" trajectories of a stationary ergodic random sequence with the sample size $n$ can be expressed  by the formula  $N\sim\exp(nH_s)$, when $n\to\infty$. Here, $H_s$ is the Shannon entropy of the underlying distribution (see details in, e.g. \cite{ref1}).
Hence,  the higher the entropy the more "complicated" the system is.

 Kolmogorov and Sinai (see, e. g., \cite{ref2}) introduced the   concept of entropy in the theory of dynamical systems.  In fact, their definition was  a generalization of the Shannon entropy.  The dynamical system's entropy  is  determined by the large-time asymptotic behavior of the coefficient  appearing in the logarithm of the number of different types of   trajectories of a dynamical system.   Again, the entropy of a dynamical system may serve as a measure of its "complexity":  the more "complex"  the  system, the richer the variety  of its trajectories.

Since   functions  are some  of the most basic mathematical objects,  the question of how to define complexity of a continuous function is quite natural.
It is also important for  practical applications. In particular, the quantitative characterization of complexity of a  continuous function  could be used to solve the   problem of data segmentation.  Consider, for example, the time series generated by different and unknown mechanisms (either stochastic, or  deterministic; we shall call such data \emph{non-homogeneous}). To analyze and model non-homogeneous data  it is necessary  to perform their segmentation first.

In order to estimate complexity of a continuous function, one can try to use the  Shannon entropy approach, but from our point of view, this approach is not suitable.
Indeed, let us consider  the function  $x(t)=t, \quad t\in [-1,1]$. Obviously, the distribution of the values of this function is uniform on the interval $[-1,1]$ . Therefore, formally calculated,   Shannon entropy for a discrete distribution of the function's values  on a uniform grid is maximal.  Hence, the complexity is maximal if it is measured by the entropy. But, in fact, a straight  line is a  very simple object, which is completely defined by two points.

From another point of view,   using the concept of entropy of a   dynamical system    is also inappropriate if one wants to estimate the  complexity of a continuous function.
 Indeed,  in the modern theory of dynamic systems it is assumed that their law of evolution  does not change over time. However,   non-autonomous ordinary differential equations do not satisfy this condition, and estimation of the complexity of continuous functions generated by these equations  is not covered by the theory of dynamical systems. Moreover, not every continuous function is generated by a dynamical system.

So,  to the best of our knowledge, the existing complexity theory provides no satisfactory method to estimate the complexity of a continuous function. Moreover,
from our point of view it is essential that   the proper definition of the complexity of a function  should not depend  the mechanism generating a function.

In the middle  of 1960s, Kolmogorov  suggested an algorithmic
approach to the notion of object's complexity. The main idea of this
approach (see \cite{ref3}) is as follows:  {\it  A "complex" object  requires a lot of
information for its reconstruction  and,   for a "simple" object,    little information is needed}.
He  formalized this idea in the language of the  theory of algorithms.
In particular, the algorithmic complexity measures the length of the program  which leads to the selection of a particular object from a set of objects.
This   approach   is  closest  to our definition of the  complexity of a continuous function.

The first-named author, see Darkhovsky \cite{ref4},   proposed to measure the \emph{$\epsilon$-complexity} of a continuous function by  the number of its values (given on uniform grid) which are  required to its reconstruction  by fixed family of approximation methods with a given marginal error $\epsilon$.  This approach  was successfully  pre-tested on the human electroencephalographic data \cite{ref4}.

In this paper, we further develop and modify this concept. \emph{The main result of this paper}
provides an effective characterization of the $\epsilon$-complexity for  a class of continuous functions given on a unit cube in the finite dimensional Euclidean space.
Specifically, we prove that, for  H\"older class functions,  \emph{there exists an affine relationship between the $\epsilon$-complexity and the logarithm  of the function reconstruction  error $\epsilon$.
}
 In other words, the $\epsilon$-complexity of the  H\"{o}lder class functions can be characterized by a pair of real numbers.

The above result, leads us to   formulate the following conjecture:
  \emph{
The $\epsilon$-complexity of an individual function from the H\"older class also has, in logarithmic coordinates,  an affine dependence   on $\epsilon$,  and
   also  can be characterized by a pair of real numbers.}   

%

This conjecture is supported by preliminary simulations.

The paper is organized as follows. In Section 2, we propose a definition of  $\epsilon$-complexity of a continuous function given on a unite cube in the finite dimensional Euclidean space. In Section 3, we give a definition
 of the $\epsilon$-complexity for a class of functions with a fixed modulus of continuity, and prove the theorem regarding the  $\epsilon$-complexity of a  functional class. The corollary of this theorem gives a characterization of the $\epsilon$-complexity for the H\"older class of functions. In this section, we also formulate  the  conjecture which characterized the $\epsilon$-complexity of an individual function from the H\"older class.

In Section 4, we introduce a definition  of the $\epsilon$-complexity of an individual continuous function given by its values on a discrete grid. In Section 5, we discuss the computational aspects of $\epsilon$-complexity's evaluation,  and formulate  our basic conjecture. This conjecture gives a numerical characterization of the $\epsilon$-complexity. 
Finally,  conclusions are provided in Section 6.


\section{Complexity  of an individual   function}
Without loss of generality we assume that a continuous function $x(t)$ is defined on a unit cube $\mathbb{I}$ in the  space $\mathbb{R}^k$ . On the set of such functions we introduce a norm $\|\cdot \|$. To be able to compare   complexity of different functions, it is reasonable to assume that $\|x(t)\|$=1, (i.e., essentially,  to consider   $x(t)/\|x(t)\|$ instead of $x(t)$).

Let $\mathbb{Z}_h$ be  a $k$-dimensional grid with spacing  $h$, and
$\mathbb{I}_h=\mathbb{I}\cap\mathbb{Z}_h$. Assume  that we  only know    the values  of $x(t)$ at  points of the set $\mathbb{I}_h$. Given this information, with what precision can we  reconstruct the function $ x(t)$?

Suppose we  have  a fixed set  of approximation  methods $\mathcal{F}$ of functions  with values given only at  the points of    $\mathbb{I}_h$.
Let  $\hat{x}(t)$ be an approximation which is  constructed using one of the  allowable methods of approximation. Consider the approximation error
$$
\delta(h)=\inf_{\mathcal{F}}\|x(t)-\hat{x}(t)\|,
$$
where  the infimum is taken over the whole set  $\mathcal{F}$ .

It is clear that  the  function $\delta(h)$  is nondecreasing: the increase  of the  grid spacing means that we discard more and more information about the function values.
If we  fix a certain "acceptable" (user-specified) error level,  $\epsilon\ge 0$, then we can determine  the  fraction of the
 function values that could be discarded while still permitting reconstruction of  the original function (again, via  the  fixed family $\cal F$ of approximation methods) with error not exceeding $\epsilon$.
Note that, in general,  the  approximation error should be related  to the norm of the function  but,  since we assume that the  function is  normalized, $\delta (h)$   really measures the  relative error.

Let
$$
 h^*(\epsilon)=\left \{\begin{array}{ll}
                        \inf\{h\le 1:\delta(h)>\epsilon\}, & \text {if}\quad \{h:
                        \delta(h)>\epsilon\} \ne \emptyset \\
                        1,& \text{if the set is empty}
                      \end{array}
                      \right.
\eqno(1)
$$
Hence, $h^*(\epsilon)$ is the minimum grid
spacing  guaranteeing
that the error of the function reconstruction from its
values on the grid exceeds a given $\epsilon$.

The value $(1/h^*(\epsilon))^k$ estimates the number of points in the set $\mathbb{I}_{h^*(\epsilon)}$ that must be   retained to achieve    a given approximation error, and it is natural to use the quantity $ {1}/{h^*(\epsilon)}$ to define   \emph{a measure of function complexity}.
There is some flexibility here since
as a quantitative measure of the $\epsilon$-complexity we can employ any monotonically increasing function of
 $1/h^*(\epsilon)$. However, as we shall see below,  use of the logarithmic function enables us to get a particularly effective characterization of the complexity. Thus we introduce  the following definition:

\begin{definition}
The number
$$
S(\epsilon, \mathcal{F}, \| \cdot \|)\eqdef S(\epsilon)=\log \frac{1}{h^*(\epsilon)}
$$
is called  $(\epsilon, \mathcal{F}, \|\cdot\|)$-complexity (or, briefly, $\epsilon$-complexity) of an individual function $x(t)$ .
\end{definition}

In other words, $\epsilon$-complexity of a continuous function
on a segment is the (logarithmic) fraction of the function
values that must be retained to reconstruct the function
via a certain fixed family of approximation methods with
a given error.

Note that   $\epsilon$-complexity is a continuous functional on  the  space of continuous functions equipped with  the  norm which was used to define  the approximation error.

It is natural to assume that $\mathcal{F} $ contains  at least the method of approximation of functions  via affine functions of the form $at+b$. In this case,  if $x(t)$ itself is an affine function on $\mathbb I^k$, then its error-free recovery   requires knowledge of  $(k+1)$ of its values on linearly independent points.  But
$\sharp (\mathbb{I}_h)\ge(k+1)$ for any  $0<h\le 1$.
Therefore, according to Definition 1,  for any affine function we have $h^*(0)=1$, and its  $0$-complexity $S(0)$   is equal to zero.

Note,  that the  proposed  measure of   complexity  is an individual characteristic of a particular function, rather than of a set of functions generated by a certain  mechanism (as is the case   of the entropy of a dynamical system).
Furthermore, this measure does not depend on the mechanisms generating the function. It is insensitive to whether the function is a sample path of a random field/process, or a trajectory of a dynamical system.


\section{Complexity of a functional class}
Let $C$ be  the space of continuous functions  with the standard norm,  $\|x(\cdot)\|_{\scriptscriptstyle C}=\maxl_{t\in\mathbb{I}}|x(t)|$. Denote by
$$
\omega_x(h)=\max_{(t,s)\in \mathbb{I}, \|t-s\|\le h}|x(t)-x(s)|
$$
the modulus of continuity of the function $x (t)$. It is well known that the   function $\omega(\cdot)$ is continuous and non-decreasing.

Let $U\subset X_{\omega}$ be an arbitrary bounded set in the class  $X_{\omega}$  of all functions with a given   modulus of continuity $\omega(\cdot)$,  and let $R\df\supl_{x(\cdot)\in U}\|x(\cdot)\|$.
We define the $\epsilon$-complexity, $S_{cl}^{U}(\epsilon,\omega)$, of the set $U$
as follows:

\begin{definition}
The number
$$
S_{cl}^{U}(\epsilon, \omega)=\frac{1}{R}\log \big(1/h(\epsilon)\big),
$$
where $h(\epsilon)$ is the the grid spacing  such that
      the maximum (over the set  $U\subset X_{\omega} $)   error of the optimal  function reconstruction using its values on the grid does not exceed $\epsilon$,
      is called $\epsilon$-complexity of the set $U\subset X_{\omega}$.
\end{definition}

Thus, to estimate $S_{cl}^U(\cdot)$ we have to find the minimum of the maximal (over all functions from $U\subset X_{\omega}$) error of the function reconstruction from
a given class using its values on the grid with spacing $h$  (we call the corresponding error the \emph{minimax reconstruction error}).

\medskip

\begin{rem}
For any individual function it is natural to calculate the relative error of the function reconstruction, i.e. the error which is scaled to the norm of the function. But, for any bounded set from the class of continuous functions, the relative reconstruction error for the class must be calculated as the ratio of the minimax reconstruction  error to the maximal norm of the functions from the given set. Therefore, to calculate $S_{cl}^U(\cdot)$ we have to  consider the absolute minimax reconstruction error.
\end{rem}

\begin{theorem}
Let us assume that  the reconstruction error is  measured in the uniform norm $\|\cdot\|_{\scriptscriptstyle C}$, and that the modulus of continuity $\omega(\cdot)$ has  the inverse (i.e., it id  strictly  increasing).
Then, the complexity $S_{cl}^U(\cdot)$ of any bounded set $U\subset X_{\omega}$ is expressed by the following relationship:
$$
S_{cl}^U(\epsilon,\omega)=\frac{1}{R}\log\frac{\sqrt{k}}{2\omega^{-1}(\epsilon)}
\eqno(2)
$$
\end{theorem}

\medskip

\begin{rem}
If $\omega(\cdot)$ is not strictly increasing then its inverse,  $\omega^{-1}(\cdot)$,   in the formula (2) should be replaced by the generalized inverse $\min\{h:\omega(h)=\epsilon\}$ \end{rem}

\medskip

\begin{proof}
To prove the theorem, we have to calculate a minimax reconstruction error $\delta_{cl}(h)$ for any given grid size $h$. Since, we consider the norm $\|\cdot\|_{\scriptscriptstyle C}$ it is sufficient to find the value $\delta_{cl}(h)$ only for one cell from $\mathbb{I}_h$.

Let $t^0\in\mathbb{I}_h,\,t^0=(t_1,\dots, t_k), \,\,e_i$ be a $k$-dimensional vector, whose components represent an arbitrary set of zeros and ones, i.e.
$e_1=(0, 0,\dots, 0),\dots ,e_m=(1,1,\dots,1)$  (obviously, the number $m$ of such components is equal to $2^k$). Consider the  values of some function $x(t)$ from $U\subset X_{\omega}$ on a single cell of $\mathbb{I}_h$, i.e., on the set $A_t^0\eqdef \{x(t^i)\}_{i=1}^{m}$, $t^i=t^0+he_i, \quad t^0\in \mathbb{I}_h$, and pose the problem of estimating the value of the function at an arbitrary point $\tau$ inside the cell. In other words,  we have to 
we have to solve the problem
$$
\supl_{\{x(0), x(h)\}}\supl_{x(\tau)\in U}|x(\tau)-u|)\longrightarrow \inf_{u\in \mathbb{R}}
\eqno(3)
$$
 where the internal supremum is taken over all the values of $x(\tau) \in U$, and the external supremum is taken over all admissible (in  $U$) values $\{x(t^i)\}$.

Denote by $u=\varphi (\tau)$ the value  of the optimization problem (3).  By definition, the norm of $\varphi (\tau)$ is equal to the minimax reconstruction  error $\delta_{cl}(h)$.
Let $r_i:=\|\tau-t^i\|, \,\i=1,\dots, m$. Then the set of all possible values of $x(\tau)\in U$, given  the \emph{fixed collection} $\{x(t^i)\}$, is equal to the segment
$$
D\df\bigcap_{i=1}^m [x(t^i)-\omega(r_i), x(t^i)+\omega(r_i)].
$$
and the   solution of the optimization problem
(3)  under the same conditions (i.e., the minimax estimate of the function value at point $\tau$ under a given   \emph{fixed collection} of admissible (in $U$) values $\{x(t^i)\}$ ) is the midpoint of this segment, and  the error of the approximation is equal to half the length of $D$.

It is easy to see that the length of $D$ is maximal if $x(t^i)=a=const,\,i=1,\dots,m$. Then the optimal selection  in (3) is $u_{\text{opt}}= a$,  and the value  $\varphi(\tau)= \minl_{1\le i\le m}\omega(\|\tau-t^i\|)$  does not depend on $a$.

To calculate the norm (in the space of continuous functions) of the minimax recovery error it is necessary to find the "worst" point $\tau$ in the cell, that is, a point where the function $\varphi(\tau)$ reaches its maximum.
Since $\omega(\cdot)$ is a monotonically increasing function, it is necessary to find a point $\tau$ inside of the cell such that the minimum of the distances from this point to the vertices of the cell will be the highest. It is easy to see   that such a point is the center   of the cell
$\tau^*$.
Obviously, $\|\tau^*-t^i\|=\sqrt{k}h/2$. Therefore, $\|\varphi(\cdot)\|_{\scriptscriptstyle C}=\varphi(\tau^*)=\omega(\sqrt{k}h/2)=\delta_{cl}(h)$.
Finally, to find  $h(\epsilon)$ from Definition 2 we have to solve  the equation $\omega(\sqrt{k}h/2)=\epsilon$.  The solution is  $h(\epsilon)= {2\omega^{-1}(\epsilon)}/{\sqrt{k}}$  which concludes the proof of the Theorem.
\end{proof}

\begin{cor}
 The $\epsilon$-complexity of any bounded subset $U$ of the  H\"older class functions is given  by the formula
$$
S_{cl}^U(\epsilon, H)=A+B\log \epsilon,
\eqno(4)
$$
 for some values of the coefficients $A$ and $B$.
\end{cor}
\begin{proof}
By definition, for the  H\"older class functions $\omega(h)=Lh^{p}$. Therefore, $\omega^{-1}(\epsilon)=(\frac{\epsilon}{L})^{1/p}$ and we get (4) from (2).
\end{proof}

Let $x_0(\cdot)$ be an \emph{individual function} from the H\"older class $X_{H}$. Consider the set $U=\{x(\cdot)\in X_{H}: \|x(\cdot)\|_{\scriptscriptstyle C}\le\|x_0(\cdot)\|_{\scriptscriptstyle C}\}$.
Then, in the case of a sufficiently rich set  $\mathcal{F}$ of approximation methods, the $\epsilon$-complexity of an \emph{individual function} $x_0(\cdot)$  should be smaller than the $\epsilon$-complexity of the corresponding set, i.e.,  $ 0\le S_x(\epsilon,H)\le S_{cl}^U(\epsilon,H)$.
This fact justifies
the following conjecture:

\begin{hypothesis}
The $\epsilon$-complexity of an individual
function from the H\"older class satisfies (4) for some values of the coefficients $A$ and $B$.
\end{hypothesis}

\begin{rem}
It can be shown that relation of type (4)  also holds  if the error is measured in the norm of the space $L_p$.
\end{rem}


\section{Complexity of a continuous function given on a discrete grid}
In the majority of applications, we  deal with functions
given by their values at a discrete set of points (i.e., by a finite sample).
We still assume that this set of values is the trace of a continuous function on a lattice in the unit
cube of the $k$-dimensional Euclidean space. Let us consider how the definition of complexity
has to be adjusted to this situation.


Let $N^k$ be the number of values of the continuous function $x(t)$ on the $k$-dimensional lattice of the unit cube.
Consider  the quantity  $h^{*}(\epsilon)$ introduced in (1) and suppose that $[h^*(\epsilon)N]^k\gg 1$. It is easy to see that we can discard $[h^*(\epsilon)N]^k$ function values from each $k$-dimensional cube with the size $h^*(\epsilon)$, and  the  reconstruction error will be less or equal $\epsilon$.  In other words, the number of values sufficient for the function reconstruction with a relative error not exceeding $\epsilon$ is equal to $n^*=[N^k/[h^*(\epsilon)N]^k]$.

Hence, in accordance with the general idea of section 2.1, the $\epsilon$-complexity is a logarithmic fraction of $n^*$ and we can formulate the following definition
\begin{definition}
The value
$$
S_{\scriptscriptstyle N}(\epsilon)=\log \frac{N^k}{[ h^*(\epsilon)N]^k}
\eqno(5)
$$
is called $\epsilon$-complexity of the individual function $x(t)$, given by the set of its discrete values.
\end{definition}

The next result follows directly from (5).
\begin{theorem}
$$
\lim_{N\to\infty}S_{\scriptscriptstyle N}(\epsilon)=S(\epsilon)
$$
\end{theorem}

The growths of  $N$ means the growths  in the sampling frequency if the function domain is fixed. Therefore, in the case of sufficiently high sampling frequency of the function,  the $\epsilon$-complexity of the sample  calculated over the discrete set of values   is not very different from the true $\epsilon$-complexity.

Of course, the question arises what should be the sampling frequency to make this difference is quite small, but if we are
dealing with the data obtained with the same sampling frequency , this question is not essential.
In any case we must bear in mind that the comparison of functions in the case of discrete set of values can be performed only when the sampling frequency is the same.

Given the above,  we can formulate the conjecture that for the H\"older class functions
(compare with (4)) the following equality should be true
$$
S_{\scriptscriptstyle N}(\epsilon)=\mathcal{A}+\mathcal{B} \log \epsilon
\eqno(6)
$$

\section{Estimation of the complexity coefficients. Basic Conjecture.}
When processing real data,  we usually have to deal with functions defined by their
values in a discrete set of points. Therefore the algorithm to estimate  complexity is focused on this situation.

Suppose we are given an array of size $N$ of function values. Let  us
choose a number $0<\mathbb{S}<1$, and  discard from the array
$[(1-\mathbb{S})N]$  values.  In the next step we use the remaining
    $[\mathbb{S}N]$ values to approximate the  values of the function for all discarded points using a  collection $\mathcal{F}$   of approximation   methods, and find the best
approximation (the approximation with the smallest error).

Two factors have to be taken into account.
First,  the remaining  points
should be distributed relatively uniformly. Second, since the error of the
approximation
 depends on the  location of the
remaining points,  for the sake of the stability of the method it is expedient,
  for a given percentage of removed points,  to choose different selection schemes
   and average the corresponding minimal approximation errors over them.
 This will allow us to smooth out the unavoidable random errors in the calculations.

Thus, for given values of $\mathbb{S}$ we  determined  the value of
minimal error $\epsilon$  of the function recovery.
It is obvious that for any $\mathbb{S}> 0$ the error of the function recovery tends to zero as $N \to \infty$ (we always assume that the grid is  uniform). On the other hand, if the sample size $N$ is too small, then   estimation  of the   recovery error  will be affected by  calculations errors
even for  values of $\mathbb{S}$ close to 1.

For this reason and based on the previous one (see (4),(6)), we can state the following \emph{basic conjecture}:

\begin{hypothesis}
For any function from the H\"older class given by its discrete values, we can specify a sample  size  $N$ of the data such that with this size there exists  an interval $[\alpha; \beta]$, $0<\alpha \le \mathbb{S} \le \beta < 1$. Within this interval the following relationship holds:
$$
\log \epsilon=\mathbb{A}+\mathbb{B} \log\mathbb{S},
\eqno(7)
$$
where $\epsilon$ is the minimal error of the function recovery by given set of reconstruction methods.
\end{hypothesis}

Let us explain the relationship (7). According to the definition, the $\epsilon$-complexity is a logarithm  of the number of function
values needed to reconstruct the function with the  error
 $\epsilon$.
 Therefore, according to (4) and (6) we take log only for  $\epsilon$.
 In the case of discrete data we deal with the value $\mathbb{S}$ and analogy of the $\epsilon$-complexity in that case is  $\log\mathbb{S}$.

Our  preliminary results of computational experiments
show
that the relationship (7) holds fairly well. The
description of the computational experiments and
simulations are in preparation and will be presented in a separate publication.


\begin{rem}
It follows from the main hypothesis that there exists a correspondence between any H\"older function and the parameters $(\mathbb{A},\mathbb{B})$ of its $\epsilon$-complexity. But this correspondence is not one-to-one. Thereupon there is a question whether it is possible to distinguish between the functions with the nearest parameters $(\mathbb{A},\,\mathbb{B})$? It is useful to consider the discrete analogues of the derivatives (i.e., the corresponding differences of the order $i,\,i=1,\dots,s)$. These analogues can be obtained from the initial sample and then for these differences it is necessary to find complexity parameters $\{\mathbb{A},\,\mathbb{B}\}_{i=1}^s$. These additional parameters will improve distinguishability of functions with the close complexity parameters.
\end{rem}

\section{Conclusion}

 In this paper we proposed a formal definition of  the $\epsilon$-complexity of a continuous function   defined on a unite cube in  a finite-dimensional space. This definition is agreed with the idea of  Kolmogorov complexity of objects. Roughly speaking, the $\epsilon$-complexity of a continuous function  can be estimated by the  fraction of the function   values which  is required to reconstruct the function  with given error $\epsilon$  and with given set of approximation methods.

  We show that the $\epsilon$-complexity has an effective   characterization, due to the detected {\it affine dependance}: the $\epsilon$-complexity of an individual function of the H\"older class  can be characterized by the pair of real numbers which we called the \emph{complexity  coefficients}.

 It has potential to be used for the  problem of segmentation of  time series and classification problem.
  All known methods of non-homogeneous data segmentation are based on information about changing  probabilistic distributions (in case of probabilistic generating mechanisms) or  models of generating mechanisms (in case deterministic or mixed mechanisms).
 If  the  time series is generated by different mechanisms (either  probabilistic, or deterministic)  in different time intervals,  complexity coefficients  can be used as  "internal"  characteristics of the function.
    Therefore it will enable us to detect  changes of data generating mechanisms \emph{using only the "internal" characteristics of a function (i.e., $\epsilon$-complexity)}.

\end{document}